\documentclass[journal]{IEEEtranTIE}
\usepackage{graphicx}
\pdfoutput=1
\usepackage{cite}
\usepackage{picinpar}
\usepackage{amsmath}
\usepackage{url}
\usepackage{flushend}
\usepackage[latin1]{inputenc}
\usepackage{colortbl}
\usepackage{soul}
\usepackage{multirow}
\usepackage{pifont}
\usepackage{color}
\usepackage{alltt}
\usepackage[hidelinks]{hyperref}
\usepackage{enumerate}
\usepackage{siunitx}
\usepackage{breakurl}
\usepackage{epstopdf}
\usepackage{pbox}
\usepackage{subfigure}

\newtheorem{assumption}{Assumption}
\newtheorem{lemma}{Lemma}
\newtheorem{remark}{Remark}
\newtheorem{definition}{Definition}
\newtheorem{theorem}{Theorem}
\newtheorem{proposition}{Proposition}
\newenvironment{proof}{{\indent \indent \it Proof:}}

\begin{document}
\title{	Adaptive Fast Smooth Second-Order Sliding Mode Control for Attitude Tracking of a 3-DOF Helicopter}

\author{
	\vskip 1em
	
	Xidong Wang, Zhan Li, Zhen He, Huijun Gao

	\thanks{
					
		Xidong Wang is with the Research Institute of Intelligent Control and Systems, School of Astronautics, Harbin Institute of Technology, Harbin 150001, China (e-mail: 17b904039@stu.hit.edu.cn). 
	}
}

\maketitle
	
\begin{abstract}
This paper presents a novel adaptive fast smooth second-order sliding mode control for the attitude tracking of the three degree-of-freedom (3-DOF) helicopter system with lumped disturbances. Combining with a non-singular integral sliding mode surface, we propose a novel adaptive fast smooth second-order sliding mode control method to enable elevation and pitch angles to track given desired trajectories respectively with the features of non-singularity, adaptation to disturbances, chattering suppression and fast finite-time convergence. In addition, a novel adaptive-gain smooth second-order sliding mode observer is proposed to compensate time-varying lumped disturbances with the smoother output compared with the adaptive-gain second-order sliding mode observer. The fast finite-time convergence of the closed-loop system with constant disturbances and the fast finite-time uniformly ultimately boundedness of the closed-loop system with the time-varying lumped disturbances are proved with the finite-time Lyapunov stability theory. Finally, the effectiveness and superiority of the proposed control methods are verified by comparative simulation experiments.
\end{abstract}

\begin{IEEEkeywords}
Adaptive fast smooth second-order sliding mode control (AFSSOSMC), adaptive-gain smooth second-order sliding mode observer (ASSOSMO), 3-DOF helicopter.
\end{IEEEkeywords}

{}

\definecolor{limegreen}{rgb}{0.2, 0.8, 0.2}
\definecolor{forestgreen}{rgb}{0.13, 0.55, 0.13}
\definecolor{greenhtml}{rgb}{0.0, 0.5, 0.0}

\section{Introduction}

\IEEEPARstart{D}{ue} to the merits of vertical take-off and landing, air hovering as well as aggressive maneuver, the small unmanned helicopter has an extremely broad application prospect in military and civil fields\cite{RBFNN2015}. However, small unmanned helicopters have the characteristics of high nonlinearity, strong coupling, under-actuated and extremely vulnerable to lumped disturbances during flight, which make it a challenge to design the high-performance attitude tracking controller \cite{{Li2015F}}.

Because of the similar dynamics with the real helicopter system, the 3-DOF laboratory helicopter, as shown in Figure 1, can act as an ideal experimental platform for testing various advanced control methods of the helicopter \cite{zheng2011}. In recent years, researchers have proposed numerous methods to achieve the attitude tracking object of 3-DOF helicopter and verified these methods by the helicopter experimental platform. These methods can be mainly divided into three parts: linear control \cite{liu2013L,Boby2014L}, nonlinear control \cite{Castaneda2016F,Kara2019F,Li2015F,Vazquez2017F,Yang2020F} and intelligent control \cite{chen2016N,Chen2018N,Zeghlache2017N}.
\begin{figure}[!t]\centering
	\includegraphics[width=8.5cm]{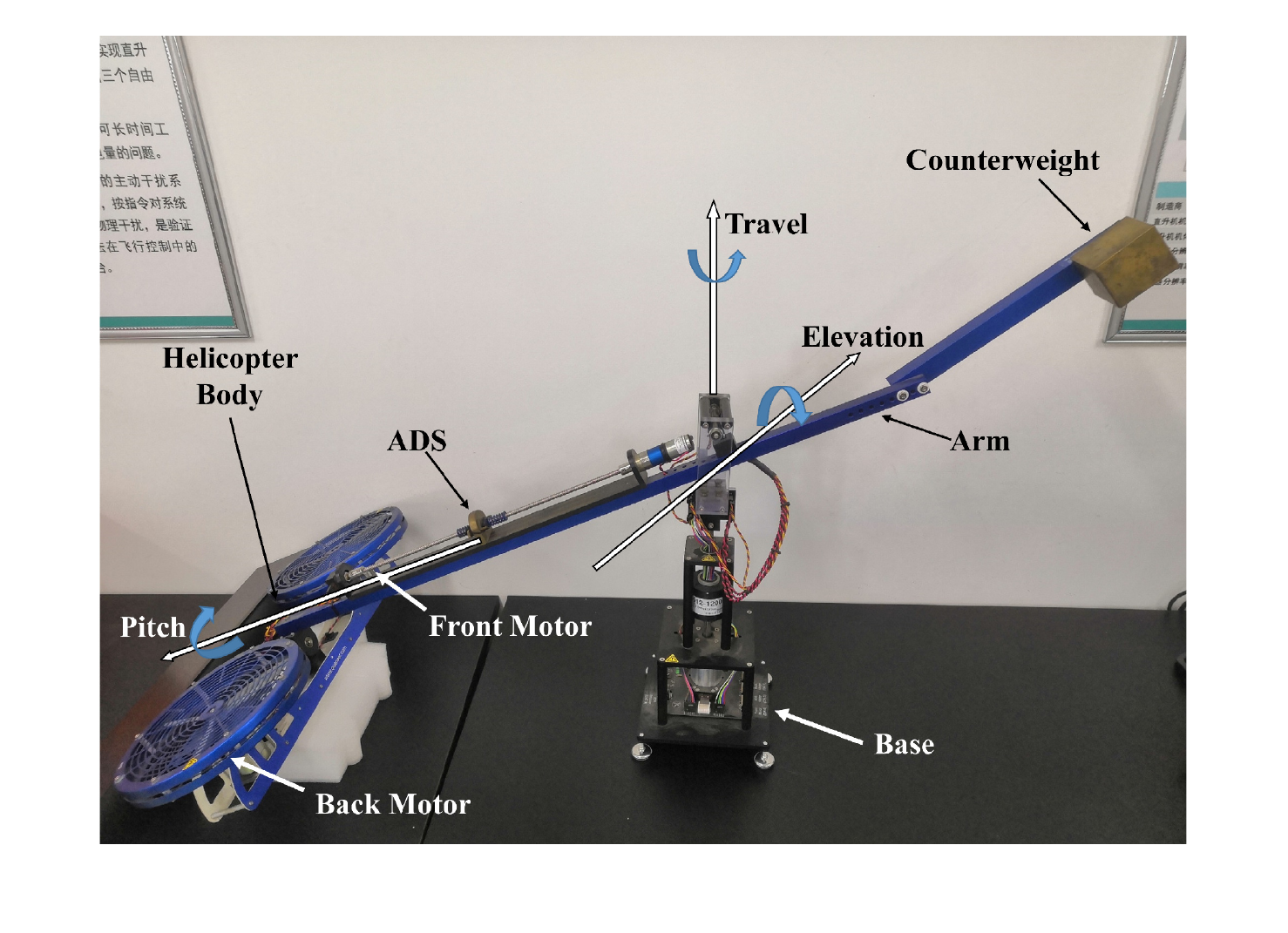}
	\caption{Structure of 3-DOF helicopter system with ADS}\label{FIG_1}
\end{figure}

Among all the above-mentioned control methods, sliding mode control has attracted much attention because of its insensitivity and strong robustness to the disturbances. In the present sliding mode algorithms, the super-twisting algorithm is very popular and owns practical application value due to the features of finite-time convergence, strong robustness and solely requiring the information of sliding mode variables \cite{Levant2003}. In \cite{Moreno2008}, a fast super-twisting algorithm is proposed to solve the problem that the convergence speed becomes slow when the system states are far away from their equilibrium points existing in \cite{Levant2003}. However, the bound of disturbance needing to be known in advance restricts the applications of these super-twisting algorithms. To deal with this problem, an adaptive super-twisting algorithm is proposed in \cite{Shtessel2012}, which can adapt to the disturbance of unknown boundary. Combining the merits of \cite{Moreno2008} and \cite{Shtessel2012}, an adaptive fast second-order sliding mode control method is proposed in \cite{ASTC2015}. An observer is also designed in the light of this method, but the output of the proposed observer is not smooth. In \cite{2007Smooth}, a smooth second-order sliding mode control is proposed to further alleviate the chattering effect existing in \cite{Levant2003}. In \cite{Jiang2018}, a fast smooth second-order sliding mode method is present based on the method of \cite{2007Smooth}. However, this method cannot adapt to the unknown boundary disturbances. Considering all the above-mentioned modified super-twisting methods, there is no method that can simultaneously achieve fast finite-time convergence, adaptation to the disturbances and smooth as of now.

In this paper, inspired by \cite{ASTC2015,Jiang2018}, we propose a class of smooth adaptive fast second-order sliding mode control (SAFSOSMC) method and the method in \cite{ASTC2015} can be regarded as a special case of our proposed method to some extent. The new proposed method integrates the advantages of all the modified super-twisting algorithms mentioned before, including fast finite-time convergence, adaptation to the disturbances and smooth. To the best of the authors' knowledge, this is the first time that a modified super-twisting algorithm can simultaneously own the above three merits. Based on this new method, a novel adaptive-gain smooth second-order sliding mode observer (ASSOSMO) is also proposed to alleviate chattering effect and adjust the parameters automatically without knowing the boundary of the lumped disturbances derivative. The main contributions of this paper can be summarized as follows:
\begin{enumerate}[1)]
	\item A novel SAFSOSMC method is proposed to enable elevation and pitch angles to track given desired trajectories respectively, which not only maintains the characteristics of fast finite-time convergence and adaptation to the disturbance of unknown boundary, but also enormously reduces the chattering effect. Combined with a non-singular integral sliding mode surface\cite{2002Surface}, this method can ensure the fast finite-time convergence of the tracking errors with constant disturbances and the fast finite-time uniformly ultimately boundedness with the time-varying lumped disturbances.
	\item A novel ASSOSMO is proposed to estimate the time-varying lumped disturbances of helicopter system control channels with smoother output than that of the adaptive-gain second-order sliding mode observer (ASOSMO) proposed in \cite{ASTC2015}.
\end{enumerate}

The fast finite-time convergence of the closed-loop system with constant disturbances and the fast finite-time uniformly ultimately boundedness of the closed-loop system with time-varying lumped disturbances will be proved with the corresponding finite-time Lyapunov stability theory. The effectiveness and superiority of the proposed control methods are verified by comparative experiments.

The remainder of this paper is organized as follows. In Section 2, the dynamic model and control objective of the 3-DOF helicopter system, some essential lemmas and the new proposition are given. The controller design process and stability analysis of the closed-loop system are presented in Section 3. Section 4 provides contrastive experiment results and discussion. Section 5 concludes this paper.

Notation: In this paper, we use $\left\| \cdot \right\|$ to denote the Euclidean norm of vectors and ${\mathop{\rm sign}\nolimits} ( \cdot )$ to denote the standard signum function. Moreover, ${\lambda _{\max }}\left(  \cdot  \right)$ and ${\lambda _{\min }}\left(  \cdot  \right)$ are used to denote the maximum and minimum eigenvalues of a matrix, respectively.
\section{Problem formulation and preliminaries}
\subsection{The dynamics of the 3-DOF Helicopter}
As shown in Fig.1, the 3-DOF helicopter system studied in this paper has elevation, pitch and travel motions, which are driven by two DC motors called the front motor and back motor. A positive voltage applied to each motor can generate the elevation motion, and a higher voltage applied on the front motor can produce the positive pitch motion . The travel motion can be generated by thrust vectors when the helicopter body is pitching. Moreover, an active disturbance system (ADS) serving as the lumped disturbances is installed on the arm.

Due to the under-actuated mechanism of the 3-DOF helicopter system, only two of the three degree of freedoms can be controlled to track arbitrary trajectories in the operating domain. In this work, we investigate the elevation and pitch motions, while the travel motion is set to move freely. The models of elevation and pitch channels can be expressed as follows \cite{Li2015F}

\begin{equation}
\begin{aligned}\label{system1}
{J_\alpha }\ddot \alpha &= {K_f}{L_a}\cos (\beta )({V_f} + {V_b}) - mg{L_a}\cos (\alpha )\\
{J_\beta }\ddot \beta & = {K_f}{L_h}({V_f} - {V_b})
\end{aligned}
\end{equation}
where $\alpha $ and $\beta $ represent the elevation and pitch angle respectively. Taking into account the mechanical constraints, the operating domain of the helicopter system is defined as follows
\begin{equation}
\begin{aligned}
- {27.5^o} \le \alpha  \le  + {30^o}\\
 - {45^o} \le \beta  \le  + {45^o}
\end{aligned}
\end{equation}
The definitions and values of the relevant parameters are shown in Table \ref{table:1}.
\begin{table}[!t]
	\renewcommand{\arraystretch}{1.3}
	\caption{The parameters of the 3-DOF helicopter system}\label{table:1}
	\centering
	%\centering
	\resizebox{\columnwidth}{!}{
		\begin{tabular}{l l l}
			\hline\hline \\[-3mm]
			\multicolumn{1}{c}{Symbol} & \multicolumn{1}{c}{Definition} & \multicolumn{1}{l}{Value} \\ \hline
			${J_\alpha }$  & Moment of inertia of elevation axis & $ 1.0348 kg\cdot m^2 $ \\
            ${J_\beta }$  &  Moment of inertia of pitch axis & $0.0451 kg\cdot m^2 $ \\
			\pbox{20cm}{${L_a}$\\\hphantom{1}} & \pbox{20cm}{Distance from elevation axis to the center \\ \hphantom{1}of helicopter body} &  \pbox{20cm}{$0.6600 m $ \\ \hphantom{1}} \\ 
			${L_h}$ & Distance from pitch axis to either motor & $0.1780 m $ \\
			${m}$ & Effective mass of the helicopter & $0.094kg $ \\
			${g}$ & Gravitational acceleration constant & $9.81m/s^2 $ \\
			${K_f}$ & Propeller force-thrust constant & $0.1188N/V $ \\
			${V_f}$ & Front motor voltage input & $[-24,24]V$ \\
			$ {V_b} $ & Back motor voltage input & $[-24,24]V $ \\
			\hline\hline
		\end{tabular}
	}
\end{table}
Denote ${x_1} = \alpha ,{x_2} = \dot \alpha ,{x_3} = \beta ,{x_4} = \dot \beta $. In view of the external disturbance and system uncertainty, the model of helicopter system can be rewritten as
\begin{equation}
\begin{aligned}
{{\dot x}_1} &= {x_2}\\
{{\dot x}_2} &= \frac{{{L_a}}}{{{J_\alpha }}}\cos ({x_3}){u_1} - \frac{g}{{{J_\alpha }}}m{L_a}\cos ({x_1}) + {d_1}(x)\\
{{\dot x}_3} &= {x_4}\\
{{\dot x}_4} &= \frac{{{L_h}}}{{{J_\beta }}}{u_2} + {d_2}(x)
\end{aligned}
\end{equation}
where $x = {[{x_1},{x_2},{x_3},{x_4}]^T}$ denotes the state vector of system \eqref{system1}, which is assumed to be measurable in this study. ${d_1}\left( x \right)$ and ${d_2}\left( x \right)$ represent the lumped disturbances existing in the corresponding control channels. In addition, ${u_1}$and ${u_2}$ are defined by
\begin{equation}
\begin{aligned}
{u_1} &= {K_f}({V_f} + {V_b})\\
{u_2} &= {K_f}({V_f} - {V_b})
\end{aligned}
\end{equation}

The control objective is to design the controllers such that the elevation and pitch angles can track the given desired trajectories respectively within small errors in finite time. For the design of the controllers, the following assumptions are required.

\begin{assumption}
The lumped disturbances and their first derivative are assumed to be bounded, while the value of the bound is unknown, i.e., there exist unknown constants ${\overline d _{1}}$, $\delta_{1}$, ${\overline d _{2}}$, $\delta_{2}$ such that $\left| {{d_1}\left( x \right)} \right| \le {\overline d _{1}}$, $\left| {{{\dot d}_1}\left( x \right)} \right| \le {\delta_{1}}$, $\left| {{d_2}\left( x \right)} \right| \le {\overline d _{2}}$ and $\left| {{{\dot d}_2}\left( x \right)} \right| \le {\delta_{2}}$, where ${\overline d _1} > 0,{\delta _1} \ge 0,{\overline d _2} > 0,{\delta _2} \ge 0$
\end{assumption}
\begin{assumption}
 The desired trajectories given by ${x_{\alpha d}}(t)$, ${x_{\beta d}}(t)$ are assumed to be bounded and available up to their second derivative. 
\end{assumption}
\subsection{Definitions and Lemmas}
To better describe the following \emph{definitions} and \emph{lemmas}, we consider a general system
\begin{equation}
\dot x = g(x(t)),{x_0} = x(0)
\end{equation}
where $g:{W_0} \to {R^n}$  is continuous on an open neighborhood ${W_0} \subset {R^n}$ of the origin and assume that $g(0) = 0$.  $x(t,{x_0})$ denote the solution of (5), which is in the sense of Filippov \cite{1999Filippov}.

\begin{definition}[\cite{1998def1}]
If the origin of (5) is Lyapunov stable and any solution of (5) converges to the equilibrium point in finite time, i.e. $\forall {x_0} \in {W_1} \subset {W_0}$, $x(t,{x_0}) \in {W_1}\backslash \left\{ 0 \right\}$ when $t \in [0,T({x_0})]$ with ${\lim _{t \to T({x_0})}}x(t,{x_0}) = 0$, and $x(t,{x_0}) = 0$ $\forall t > T({x_0})$, where $T:{W_1}\backslash \left\{ 0 \right\} \to (0,\infty )$ is a continuous function, then the system (5) is called finite-time stable.
\end{definition}

\begin{definition}[\cite{2011def2}]
If any solution of (5) converges to a neighborhood of the origin in finite time, the system (5) is called finite-time uniformly ultimately boundedness. 
\end{definition}

\begin{lemma} [\cite{2005lemma1}]
Suppose there exists a continuous and positive-definite function $V:{W_0} \to R$ satisfying the condition
\begin{equation}
\dot V(x) \le  - {c_1}V{(x)^p} - {c_2}V(x)
\end{equation}
where ${c_1} > 0,{c_2} > 0,p \in (0,1)$, then the trajectory of (5) is fast finite-time stable. Moreover, the settling time is given by:
\begin{equation}
T \le \frac{{\ln [1 + {c_2}V{{({x_0})}^{1 - p}}/{c_1}]}}{{{c_2}(1 - p)}}
\end{equation}
\end{lemma}

\begin{lemma} [\cite{Jiang2018}]
Suppose there exists a continuous and positive-definite function $V:{W_0} \to R$ satisfying the condition 
\begin{equation}
\dot V(x) \le  - {c_1}V{(x)^{{p_1}}} - {c_2}V(x){\rm{ + }}{c_3}V{(x)^{{p_2}}}
\end{equation}
where ${c_1} > 0,{c_2} > 0,{c_3} > 0,{p_1} \in (0,1),{p_2} \in (0,{p_1})$, then the trajectory of (5) is fast finite-time uniformly ultimately boundedness. Moreover, the settling time is given by:
\begin{equation}
T \le \frac{{\ln [1 + ({c_2} - {\theta _2})V{{({x_0})}^{1 - {p_1}}}/({c_1} - {\theta _1})]}}{{({c_2} - {\theta _2})(1 - {p_1})}}
\end{equation}
where ${\theta _1}$ and ${\theta _2}$ are arbitrary positive constants holding ${\theta _1} \in (0,{c_1}),{\theta _2} \in (0,{c_2})$, then $x(t,{x_0})$ can converge to a neighborhood of origin in a finite time $T$. In addition, the convergent region can be given by:
\begin{equation}
D = \left\{ {x:{\theta _1}V{{(x)}^{{p_1} - {p_2}}} + {\theta _2}V{{(x)}^{1 - {p_2}}} < {c_3}} \right\}
\end{equation}

Define an auxiliary variable ${\theta _3} \in \left( {0,1} \right)$. If ${\theta _3}$ is selected satisfying
\begin{equation}
{\theta _3}^{1 - {p_2}}{\theta _2}^{{p_1} - {p_2}}{c_3}^{1 - {p_2}} = {\theta _1}^{1 - {p_2}}{(1 - {\theta _3})^{{p_1} - {p_2}}}
\end{equation}
then (10) can be reduced to $D = {D_1} = {D_2}$, where
\begin{equation}
\begin{aligned}
{D_1}& = \left\{ {x:V{{(x)}^{{p_1} - {p_2}}} < {\theta _3}{c_3}/{\theta _1}} \right\}\\
{D_2}& = \left\{ {x:V{{(x)}^{1 - {p_2}}} < \left( {1 - {\theta _3}} \right){c_3}/{\theta _2}} \right\}
\end{aligned}
\end{equation}
which means the state $x$ can converge to ${D_1} = {D_2}$ in finite time $T$.
\end{lemma}

\begin{remark}
The meaning of fast in \emph{Lemma 1} and \emph{Lemma 2} is that the solution of (5) can quickly converge to the origin or a neighborhood of the origin regardless of the distance between the initial state and the origin, while the original finite-time stable converges slowly when the initial state is far away from the origin.
\end{remark}

\begin{lemma} [\cite{2002Surface}]
Considering the following system
\begin{equation}
\begin{aligned}
{{\dot x}_1} &= {x_2}\\
{{\dot x}_2} &=  - {\gamma _1}{\left| {{x_1}} \right|^{{p_3}}}{\mathop{\rm sgn}} ({x_1}) - {\gamma _2}{\left| {{x_2}} \right|^{2{p_3}/(1 + {p_3})}}{\mathop{\rm sgn}} ({x_2})
\end{aligned}
\end{equation}
where ${\gamma _1} > 0,{\gamma _2} > 0,{p_3} \in (0,1)$, the trajectory of the system (13) is finite-time stable, i.e. $x_1,x_2$ can converge to the origin in finite time.
\end{lemma}
%2-C
\subsection{A New Proposition}
Motivated by [18], we extend the result in [16] to obtain the following new proposition and this generalization is non-trivial.
\begin{proposition}
Considering the following system
\begin{equation}
\begin{aligned}
\dot s &=  - {L_1}(t){\left| s \right|^{\frac{{m - 1}}{m}}}{\mathop{\rm sgn}} (s) - {L_2}(t)s + \varphi \\
\dot \varphi  &=  - {L_3}(t){\left| s \right|^{\frac{{m - 2}}{m}}}{\mathop{\rm sgn}} (s) - {L_4}(t)s + d\left( t \right)
\end{aligned}
\end{equation}
where $\left| {d\left( t \right)} \right| \le \delta $, $\delta$ is an unknown non-negative constant and the adaptive gains ${L_1}(t),{L_2}(t),{L_3}(t),{L_4}(t)$ are expressed as
\begin{equation}
\begin{aligned}
{L_1}(t)&={k_1}L_0^{{\textstyle{{m - 1} \over m}}}\left( t \right),{\rm{  }}{L_2}(t) = {k_2}L_0 \left( t \right)\\
{L_3}(t)& = {k_3}L_0^{{\textstyle{{2m - 2} \over m}}}\left( t \right),{\rm{  }}{L_4}(t) = {k_4}L_0^2 \left( t \right)
\end{aligned}
\end{equation}
where ${k_1},{k_2},{k_3},{k_4},m$ are positive constants satisfying
\begin{equation}
{m^2}{k_3}{k_4} > \left( {\frac{{{m^3}{k_3}}}{{m - 1}} + \left( {4{m^2} - 4m + 1} \right){k_1}^2} \right){k_2}^2,m > 2
\end{equation}

${L_0}(t)$ is a scalar, positive, and time-varying function. The ${L_0}(t)$ satisfies
\begin{equation}
\dot{L_0}\left( t \right) =
\begin{cases}
{\kappa},&if \quad \left| s \right| \ge \varepsilon \hfill \\
0,&else \hfill \\
\end{cases} 
\end{equation}
where $\kappa$ is a positive constant, $\varepsilon$ is an arbitrary small positive value. Then we can obtain the following conclusions

(\romannumeral1) If $d = 0$, then $s,\dot s$ can fast converge to the origin in finite time.

(\romannumeral2) If $d \ne 0$, then $s,\dot s$ can fast converge to a neighborhood of the origin in finite time.
\end{proposition}

\begin{proof}
Define a new state vector
\begin{equation}
\xi  = \left[ {\begin{array}{*{20}{c}}
{{\xi _1}}\\
{{\xi _2}}\\
{{\xi _3}}
\end{array}} \right] = \left[ {\begin{array}{*{20}{c}}
{L_0^{{\textstyle{{m - 1} \over m}}}\left( t \right){{\left| s \right|}^{{\textstyle{{m - 1} \over m}}}}{\mathop{\rm sgn}} (s)}\\
{L_0 \left( t \right)s}\\
\varphi 
\end{array}} \right]
\end{equation} 

After taking the derivative of $\xi$, we can obtain
\begin{equation}
\begin{aligned}
\dot \xi  =& L_0^{{\textstyle{{m - 1} \over m}}}{\left| s \right|^{\frac{{ - 1}}{m}}}\left[ {\begin{array}{*{20}{c}}
{ - \frac{{m - 1}}{m}{k_1}}&{ - \frac{{m - 1}}{m}{k_2}}&{\frac{{m - 1}}{m}}\\
0&0&0\\
{ - {k_3}}&0&0
\end{array}} \right]\xi \\
 &+ L_0 \left[ {\begin{array}{*{20}{c}}
0&0&0\\
{{\rm{ - }}{k_1}}&{ - {k_2}}&1\\
0&{ - {k_4}}&0
\end{array}} \right]\xi  + \left[ {\begin{array}{*{20}{c}}
{\frac{{m - 1}}{m}\frac{{{{\dot L}_0}(t)}}{{{L_0}(t)}}{\xi _1}}\\
{\frac{{{{\dot L}_0}(t)}}{{{L_0}(t)}}{\xi _2}}\\
{d(t)}
\end{array}} \right]
\end{aligned}
\end{equation}

Then a candidate Lyapunov function is chosen as $V\left( \xi  \right) = {\xi ^T}P\xi$, with
\begin{equation}
\begin{aligned}
P = \frac{1}{2}\left[ {\begin{array}{*{20}{c}}
{\frac{{2m}}{{m - 1}}{k_3} + {k_1}^2}&{{k_1}{k_2}}&{ - {k_1}}\\
{{k_1}{k_2}}&{2{k_4} + {k_2}^2}&{ - {k_2}}\\
{ - {k_1}}&{ - {k_2}}&2
\end{array}} \right]
\end{aligned}
\end{equation}
where $P$ is a symmetric positive definite matrix. Taking the derivative of $V\left( \xi  \right)$ along the trajectories of system (14) yields 
\begin{equation}
	\begin{aligned}
		\dot V\left( \xi  \right) = & - {L_0}{\left| {{\xi _1}} \right|^{{\textstyle{{ - 1} \over {m - 1}}}}}{\xi ^T}{\Omega _1}\xi  - L_0 {\xi ^T}{\Omega _2}\xi  \\
		                            &+ {\sigma _1}d\left( t \right)\xi  + \frac{{2m - 2}}{m}\frac{{{{\dot L}_0}(t)}}{{{L_0}(t)}}{\sigma _2}P\xi 
	\end{aligned}	
\end{equation}
where $\sigma_1 = \left[-k_1 -k_2 \quad 2 \right]$, $\sigma_2 =\left[\xi _1 \quad \frac{m}{m-1}{\xi _2} \quad 0\right]$ and
\begin{equation}\small
\begin{aligned}
&{\Omega _1}=\\
&\frac{{{k_1}}}{m}\left[ {\begin{array}{*{20}{c}}
{{k_3}m + k_1^2\left( {m - 1} \right)}&0&{ - {k_1}\left( {m - 1} \right)}\\
0&{{k_4}m + k_2^2\left( {3m - 1} \right)}&{ - {k_2}\left( {2m - 1} \right)}\\
{ - {k_1}\left( {m - 1} \right)}&{ - {k_2}\left( {2m - 1} \right)}&{m - 1}
\end{array}} \right]\\
&{\Omega _2} = {k_2}\left[ {\begin{array}{*{20}{c}}
{{k_3} + k_1^2\left( {3m - 2} \right)/m}&0&0\\
0&{{k_4} + k_2^2}&{ - {k_2}}\\
0&{ - {k_2}}&1
\end{array}} \right]
\end{aligned}
\end{equation}

It is easy to prove that the matrices $\Omega _1$ and $\Omega _2$ both are positive definite with (16). By using
\begin{equation}
{\lambda _{\min }}\left( P \right){\left\| \xi  \right\|^2} \le V \le {\lambda _{\max }}\left( P \right){\left\| \xi  \right\|^2}
\end{equation}
(21) can be formulated as
\begin{equation}
\begin{aligned}
\dot V \le & - {L_0}\left( t \right)\frac{{{\lambda _{\min }}\left( {{\Omega _1}} \right)}}{{\lambda _{\max }^{{p_1}}\left( P \right)}}{V^{{p_1}}} - L_0 \left( t \right)\frac{{{\lambda _{\min }}\left( {{\Omega _2}} \right)}}{{{\lambda _{\max }}\left( P \right)}}V \\
&+ \frac{{\delta {{\left\| {{\sigma _1}} \right\|}_2}}}{{\lambda _{\min }^{{\raise0.5ex\hbox{$\scriptstyle 1$}
\kern-0.1em/\kern-0.15em
\lower0.25ex\hbox{$\scriptstyle 2$}}}\left( P \right)}}{V^{\frac{1}{2}}} + \frac{{m - 1}}{m}\frac{{{{\dot L}_0}}}{{{L_0}}}{\xi ^T}Q\xi 
\end{aligned}
\end{equation}
where ${p_1} = \left( {2m - 3} \right)/\left( {2m - 2} \right)$, $Q = diag\left[ {{q_1},{q_2},{q_3}} \right]$ is a diagonal matrix with positive diagonal elements, which are expressed as follows
\begin{equation}
\begin{aligned}
{q_1}& = \frac{{2m}}{{m - 1}}{k_3} + k_1^2 + \frac{{\left( {2m - 1} \right){k_1}{k_2}}}{{2\left( {m - 1} \right)}} + \frac{{{k_1}}}{2}\\
{q_2} &= \frac{m}{{2\left( {m - 1} \right)}}\left( {4{k_4} + 2k_2^2 + {k_2}} \right) + \frac{{\left( {2m - 1} \right){k_1}{k_2}}}{{2\left( {m - 1} \right)}}\\
{q_3} &= \frac{{{k_1}}}{2} + \frac{{m{k_2}}}{{2\left( {m - 1} \right)}}
\end{aligned}
\end{equation}

Then (24) can be further rewritten as
\begin{equation}\small
\dot V \le  - {L_0}\left( t \right){n_1}{V^{{p_1}}} + {n_2}{V^{\frac{1}{2}}} - \left( {L_0 \left( t \right){n_3} - \frac{{2m - 2}}{m}{n_4}\frac{{{{\dot L}_0}}}{{{L_0}}}} \right)V
\end{equation}
where
\begin{equation}
\begin{aligned}
{n_1} &= \frac{{{\lambda _{\min }}\left( {{\Omega _1}} \right)}}{{\lambda _{\max }^{{p_1}}\left( P \right)}}\\
{n_2} &= \frac{{\delta {{\left\| {{\sigma _1}} \right\|}_2}}}{{\lambda _{\min }^{{\raise0.5ex\hbox{$\scriptstyle 1$}
\kern-0.1em/\kern-0.15em
\lower0.25ex\hbox{$\scriptstyle 2$}}}\left( P \right)}}\\
{n_3} &= \frac{{{\lambda _{\min }}\left( {{\Omega _2}} \right)}}{{{\lambda _{\max }}\left( P \right)}}\\
{n_4} &= \frac{{{\lambda _{\max }}\left( Q \right)}}{{2{\lambda _{\min }}\left( P \right)}}
\end{aligned}
\end{equation}

(\romannumeral1) If $d(t)= 0$, then $\delta =0$, (26) will become
\begin{equation}
\dot V \le  - {L_0}\left( t \right){n_1}{V^{{p_1}}} - \left( {L_0 \left( t \right){n_3} - \frac{{2m - 2}}{m}{n_4}\frac{{{{\dot L}_0}}}{{{L_0}}}} \right)V
\end{equation}

Due to ${\dot L_0}\left( t \right) \ge 0$, $L_0 \left( t \right){n_3} - \left( {2m - 2} \right){n_4}{\dot L_0}/\left( {{L_0}m} \right)$ is positive in finite time. It follows
from (28) that 
\begin{equation}
\dot V \le  - {c_1}{V^{{p_1}}} - {c_2}V
\end{equation}
where ${c_1}$ and ${c_2}$ are positive constants, ${p_1} \in (0.5,1)$. By using \emph{Lemma 1}, $\xi$ can fast converge to origin in  finite time, then $s,\dot s$ can fast converge to the origin in finite time.

(\romannumeral2) If $d\left( t \right) \ne 0$, with the same analysis of (\romannumeral1), it follows from (26) that
\begin{equation}
\dot V \le  - {c_4}{V^{{p_1}}} - {c_5}V{\rm{ + }}{c_3}{V^{{\textstyle{1 \over 2}}}}
\end{equation}
where ${c_3},{c_4}$and ${c_5}$ are positive constants, ${p_1} \in (0.5,1)$. By using \emph{Lemma 2}, $\xi$ can fast converge to a neighborhood of origin in  finite time. In addition, the region can be given by
\begin{equation}
D = \left\{ {\xi :{\theta _1}V{{(\xi )}^{{p_1} - {p_2}}} + {\theta _2}V{{(\xi )}^{1 - {p_2}}} < {c_3}} \right\}
\end{equation}
where ${\theta _1} \in (0,{c_4}),{\theta _2} \in (0,{c_5})$, ${p_2} = 0.5$.

Define an auxiliary variable ${\theta _3} \in \left( {0,1} \right)$.If ${\theta _3}$ is selected satisfying
\begin{equation}
{\theta _3}^{1 - {p_2}}{\theta _2}^{{p_1} - {p_2}}{c_3}^{1 - {p_2}} = {\theta _1}^{1 - {p_2}}{(1 - {\theta _3})^{{p_1} - {p_2}}}
\end{equation}
$\xi$ can converge to $D = {D_1} = {D_2}$ in finite time ${T_1}$, where
\begin{equation}
{T_1} \le \frac{{\ln [1 + ({c_5} - {\theta _2})V{{({\xi _0})}^{1 - {p_1}}}/({c_4} - {\theta _1})]}}{{({c_5} - {\theta _2})(1 - {p_1})}}
\end{equation}

\begin{equation}
\begin{aligned}
\begin{array}{l}
{D_1} = \left\{ {\xi :V{{(\xi )}^{{p_1} - {p_2}}} < {\theta _3}{c_3}/{\theta _1}} \right\}\\
{D_2} = \left\{ {\xi :V{{(\xi )}^{1 - {p_2}}} < \left( {1 - {\theta _3}} \right){c_3}/{\theta _2}} \right\}
\end{array}
\end{aligned}
\end{equation}

Design a region $D_4$ as
\begin{equation}
\begin{aligned}
{D_4} &= \left\{ {\xi :{\lambda _{\min }}\left( P \right){{\left\| \xi  \right\|}^2} < {{\left( {1 - {\theta _3}} \right)}^2}c_3^2/{\theta _2}^2} \right\}\\
 &= \left\{ {\xi :\left\| \xi  \right\| < {\Delta}} \right\}
\end{aligned}
\end{equation}
where ${\Delta}{\rm{ = }}{\lambda _{\min }}{\left( P \right)^{ - 1/2}}\left( {1 - {\theta _3}} \right){c_3}/{\theta _2}$. In terms of (23), it follows from (34) and (35) that $D_4$ contains $D_2$. Considering the definition of $\xi$, the following inequalities $\left\| {{\xi _1}} \right\| \le \left\| \xi  \right\|$, $\left\| {{\xi _2}} \right\| \le \left\| \xi  \right\|$and $\left\| {{\xi _3}} \right\| \le \left\| \xi  \right\|$ hold. Then the set
\begin{equation}
{D_5} = \left\{ {{\xi _1},{\xi _2},{\xi _3}:\left\| {{\xi _1}} \right\| < {\Delta},\left\| {{\xi _2}} \right\| < {\Delta},\left\| {{\xi _3}} \right\| < {\Delta}} \right\}
\end{equation}
contains the set $D_4$. Therefore, since $\xi$ converges to $D_1$ in $T_1$ , it will also converge to $D_5$ in $T_1$  . Using (14), (15) and (18), we can obtain that $s,\dot s$ converge to a neighborhood of the origin in $T_1$. 
\end{proof}

\begin{remark}
By selecting different values of $m$ in (14), we can obtain a series of smooth adaptive fast second-order sliding mode control methods and if $m = 2$, (14) can be transformed into the result in [16]. However, our proposed methods can greatly reduce the chattering effect existing in [16]. The newly proposed proposition will be used in the design of controller and observer for the helicopter system in the next section, and the superiority of our proposed control method will be verified through the comparative simulation experiments.
\end{remark}

%3-A
\section{Controller and Observer Design}
In this section, we will give the design process of controller and observer in detail. First, the system (3) is transformed into the tracking error system. Then a control scheme, consisting of a non-singular integral sliding mode surface and a novel SAFSOSMC method, is proposed for the tracking error system subject to constant disturbance and time-varying lumped disturbance. Finally, a novel ASSOSMO is present to estimate disturbances.
\subsection{System transformation}
Defining the tracking errors ${e_1} = {x_1} - {x_{\alpha d}}(t),{e_2} = {x_2} - {\dot x_{\alpha d}}(t),{e_3} = {x_3} - {x_{\beta d}}(t),{e_4} = {x_4} - {\dot x_{\beta d}}(t)$. Then the tracking error system is given by
\begin{equation}
\begin{aligned}
{\dot e_1}& = {e_2}\\
{\dot e_2} &= \frac{{{L_a}}}{{{J_\alpha }}}\cos ({x_3}){u_1} - \frac{g}{{{J_\alpha }}}m{L_a}\cos ({x_1}) - {\ddot x_{\alpha d}}(t) + {d_1}(x)\\
{\dot e_3} &= {e_4}\\
{\dot e_4} &= \frac{{{L_h}}}{{{J_\beta }}}{u_2} - {\ddot x_{\beta d}}(t) + {d_2}(x)
\end{aligned}
\end{equation}

The control objective is then transformed into the design of finite-time controller so that ${e_1}$ and ${e_3}$ can converge to the origin or a small region of the origin. To facilitate the design of controllers, we express the models of the elevation and pitch channels in a unified form 
\begin{equation}
\begin{aligned}
&{\dot e_i} = {e_{i + 1}}\\
&{\dot e_{i + 1}} = {g_i}{v_i} + {f_i} + {d_i}
\end{aligned}
\end{equation}
where ${f_i}$ and ${g_i},i = 1,2$ represent the dynamics of corresponding channels, which can be express as follows
\begin{equation}
\begin{aligned}
{f_1} &=  - \frac{g}{{{J_\alpha }}}m{L_a}\cos ({x_1}) - {{\ddot x}_{\alpha d}}(t)\\
{g_1} &= \frac{{{L_a}}}{{{J_\alpha }}}\\
{f_2} &=  - {{\ddot x}_{\beta d}}(t)\\
{g_2} &= \frac{{{L_h}}}{{{J_\beta }}}
\end{aligned}
\end{equation}
and the auxiliary control inputs ${v_1}$ and ${v_2}$ are defined as
\begin{equation}
\begin{aligned}
{v_1}& = \cos ({x_3}){u_1}\\
{v_2}& = {u_2}
\end{aligned}
\end{equation}

Then, we will take the elevation channel as an example to expound the design process and the controller for the pitch channel can be designed following a similar process. Afterwards, we can get the auxiliary control inputs ${v_1}$ and ${v_2}$. By using (4) and (40), we can easily obtain the true control inputs ${V_f}$ and${V_b}$.
%3-B
\subsection{Controller design with constant disturbance}
Considering the tracking error system of the elevation channel
\begin{equation}
\begin{aligned}
&{{\dot e}_1} = {e_2}\\
&{{\dot e}_2} = {g_1}{v_1} + {f_1} + {d_1}
\end{aligned}
\end{equation}
where ${d_1}$ is constant disturbance, i.e. ${\dot d_1} = 0$. We adopt a non-singular integral sliding mode surface from [19], which is defined as follows
\begin{equation}
s = {e_2} + \int_0^t {zd\tau }
\end{equation}
where $z = {\gamma _3}{\left| {{e_1}} \right|^{{p_4}}}{\mathop{\rm sgn}} ({e_1}) + {\gamma _4}{\left| {{e_2}} \right|^{2{p_4}/(1 + {p_4})}}{\mathop{\rm sgn}} ({e_2})$ and ${\gamma _3} > 0,{\gamma _4} > 0,{p_4} \in (0,1)$.

Based on this sliding mode variable and \emph{Proposition 1}, a novel smooth adaptive fast second-order sliding mode controller can be designed as 
\begin{equation}
\begin{aligned}
{v_1} =& {g_1}^{ - 1}\left\{ { - {L_1}(t){{\left| s \right|}^{\frac{{m - 1}}{m}}}{\mathop{\rm sgn}} (s) - {L_2}(t)s - z }\right.\\
 &\phantom{=\;\;}\left.{- {f_1} - \int_0^t {\left[ {{L_3}(t){{\left| s \right|}^{\frac{{m - 2}}{m}}}{\mathop{\rm sgn}} (s) + {L_4}(t)s} \right]d\tau } } \right\}
\end{aligned}
\end{equation}
where ${L_1}(t),{L_2}(t),{L_3}(t),{L_4}(t)$ are formulated the same as (15) and $m>2$.

\begin{theorem}
Considering the tracking error system (41) subject to constant disturbance, the proposed control law (43) ensures that the tracking error can fast converge to the origin in finite time. 
\end{theorem}

\begin{proof}
We will adopt two steps to complete the proof. In \textbf{Step 1}, we will prove that the $\dot s$ can fast converge to the origin in finite time. In \textbf{Step 2}, we prove that the tracking error $e_1$ will converge to the origin in finite time when $\dot s = 0$ . 

\textbf{Step 1.} Taking the derivative of (42),we can obtain
\begin{equation}
\dot s = {\dot e_2} + z
\end{equation}

Substituting (41) and (43) into (44) leads to
\begin{equation}
\begin{aligned}
\dot s =&  - {L_1}(t){\left| s \right|^{\frac{{m - 1}}{m}}}{\mathop{\rm sgn}} (s) - {L_2}(t)s + {d_1}\\
 &- \int_0^t {\left[ {{L_3}(t){{\left| s \right|}^{\frac{{m - 2}}{m}}}{\mathop{\rm sgn}} (s) + {L_4}(t)s} \right]d\tau } 
\end{aligned}
\end{equation}

We define an intermediate variable
\begin{equation}
\varphi  =  - \int_0^t {\left[ {{L_3}(t){{\left| s \right|}^{\frac{{m - 2}}{m}}}{\mathop{\rm sgn}} (s) + {L_4}(t)s} \right]d\tau  + } {d_1}
\end{equation}

Then (45) becomes
\begin{equation}
\begin{aligned}
&\dot s =  - {L_1}(t){\left| s \right|^{\frac{{m - 1}}{m}}}{\mathop{\rm sgn}} (s) - {L_2}(t)s + \varphi \\
&\dot \varphi  =  - {L_3}(t){\left| s \right|^{\frac{{m - 2}}{m}}}{\mathop{\rm sgn}} (s) - {L_4}(t)s
\end{aligned}
\end{equation}

By using \emph{Proposition 1}, $\dot s$ can fast converge to the origin in finite time.

\textbf{Step 2.} When $\dot s = 0$, (44) becomes
\begin{equation}
{\dot e_2} =  - {\gamma _3}{\left| {{e_1}} \right|^{{p_4}}}{\mathop{\rm sgn}} ({e_1}) - {\gamma _4}{\left| {{e_2}} \right|^{2{p_4}/(1 + {p_4})}}{\mathop{\rm sgn}} ({e_2})
\end{equation}

By using \emph{Lemma 3} coupled with system (41), $e_1$ and $e_2$ can converge to the origin in finite time. The proof is completed.
\end{proof}

%3-C
\subsection{Controller design with time-varying lumped disturbance}
Considering the tracking error system of the elevation channel
\begin{equation}
\begin{aligned}
&{{\dot e}_1} = {e_2}\\
&{{\dot e}_2} = {g_1}{v_1} + {f_1} + {d_1}
\end{aligned}
\end{equation}
where ${d_1}$ is time-varying lumped disturbance satisfying \emph{Assumption 1}, i.e. $\left| {{{\dot d}_1}} \right| \le {\delta _1},{\delta _1} > 0$. We adopt the same non-singular integral sliding mode surface as (42), which is defined as follows
\begin{equation}
{s_v} = {e_2} + \int_0^t {{z_v}d\tau }
\end{equation}
where ${z_v} = {\gamma _5}{\left| {{e_1}} \right|^{{p_5}}}{\mathop{\rm sgn}} ({e_1}) + {\gamma _6}{\left| {{e_2}} \right|^{2{p_5}/(1 + {p_5})}}{\mathop{\rm sgn}} ({e_2})$ and ${\gamma _5} > 0,{\gamma _6} > 0,{p_5} \in (0,1)$.

Based on this sliding mode variable and \emph{Proposition 1}, a novel smooth adaptive fast second-order sliding mode controller can be designed as
\begin{equation}
\begin{aligned}
{v_1} =& {g_1}^{ - 1}\left\{ { - {L_{v1}}(t){{\left| {{s_v}} \right|}^{\frac{{m - 1}}{m}}}{\mathop{\rm sgn}} ({s_v}) - {L_{v2}}(t){s_v} - {z_v}}\right.\\
 &\phantom{=\;\;}\left.{ - {f_1} - \int_0^t {\left[ {{L_v}_3(t){{\left| {{s_v}} \right|}^{\frac{{m - 2}}{m}}}{\mathop{\rm sgn}} ({s_v}) + {L_{v4}}(t){s_v}} \right]d\tau } } \right\}
\end{aligned}
\end{equation}
where ${L_{v1}}(t),{L_{v2}}(t),{L_{v3}}(t),{L_{v4}}(t)$ are formulated the same as (15) and $m>2$.

\begin{theorem}
Considering the tracking error system (49) subject to time-varying lumped disturbance, the proposed control law (51) guarantees that the tracking error can fast converge to a region of the origin in finite time.
\end{theorem}
\begin{proof}
We will adopt two steps to complete the proof. In \textbf{Step 1}, we will prove that the $\dot s$ can fast converge to a neighborhood of the origin in finite time. In \textbf{Step 2}, we prove that the tracking error $e_1$ will converge to a neighborhood of the origin in finite time when $\left| {{{\dot s}_v}} \right| \le \Delta_1$.

\textbf{Step 1.} Taking the derivative of (50),we can obtain
\begin{equation}
{\dot s_v} = {\dot e_2} + {z_v}
\end{equation}

Substituting (49) and (51) into (52) leads to
\begin{equation}
\begin{aligned}
{\dot s_v} = & - {L_{v1}}(t){\left| {{s_v}} \right|^{\frac{{m - 1}}{m}}}{\mathop{\rm sgn}} ({s_v}) - {L_{v2}}(t){s_v} + {d_1}\\ 
&- \int_0^t {\left[ {{L_{v3}}(t){{\left| {{s_v}} \right|}^{\frac{{m - 2}}{m}}}{\mathop{\rm sgn}} ({s_v}) + {L_{v4}}(t){s_v}} \right]d\tau } 
\end{aligned}
\end{equation}

We define an intermediate variable
\begin{equation}
{\varphi _v} =  - \int_0^t {\left[ {{L_{v3}}(t){{\left| {{s_v}} \right|}^{\frac{{m - 2}}{m}}}{\mathop{\rm sgn}} ({s_v}) + {L_{v4}}(t){s_v}} \right]d\tau  + } {d_1}
\end{equation}

Then (53) becomes
\begin{equation}
\begin{aligned}
&{{\dot s}_v} =  - {L_{v1}}(t){\left| {{s_v}} \right|^{\frac{{m - 1}}{m}}}{\mathop{\rm sgn}} ({s_v}) - {L_{v2}}(t){s_v} + {\varphi _v}\\
&{{\dot \varphi }_v} =  - {L_{v3}}(t){\left| {{s_v}} \right|^{\frac{{m - 2}}{m}}}{\mathop{\rm sgn}} ({s_v}) - {L_{v4}}(t){s_v} + {{\dot d}_1}
\end{aligned}
\end{equation}

According to \emph{Assumption 1}, the disturbance $\left| {{{\dot d}_1}} \right| \le {\delta _1}$. By using \emph{Proposition 1}, ${\dot s_v}$ can fast converge to a neighborhood of the origin in finite time, and the region is defined as $\Delta_1$.

\textbf{Step 2.} When $\left| {{{\dot s}_v}} \right| \le \Delta_1 $, (52) becomes
\begin{equation}
{\dot e_2} + {\gamma _5}{\left| {{e_1}} \right|^{{p_5}}}{\mathop{\rm sgn}} ({e_1}) + {\gamma _6}{\left| {{e_2}} \right|^{2{p_5}/(1 + {p_5})}}{\mathop{\rm sgn}} ({e_2}) = \phi 
\end{equation}
where $\left| \phi  \right| \le \Delta_1 $.

Then (56) can be equivalently transformed into
\begin{equation}\small
{\dot e_2} + \left( {{\gamma _5} - \frac{\phi }{{{{\left| {{e_1}} \right|}^{{p_5}}}{\mathop{\rm sgn}} ({e_1})}}} \right){\left| {{e_1}} \right|^{{p_5}}}{\mathop{\rm sgn}} ({e_1}) + {\gamma _6}{\left| {{e_2}} \right|^{2{p_5}/(1 + {p_5})}}{\mathop{\rm sgn}} ({e_2}) = 0
\end{equation}

For (57), if $\gamma _5 > \left| \phi  \right|/\left( {{{\left| {{e_1}} \right|}^{{p_5}}}{\mathop{\rm sgn}} ({e_1})} \right)$, then ${\gamma _5} - \phi /\left( {{{\left| {{e_1}} \right|}^{{p_5}}}{\mathop{\rm sgn}} ({e_1})} \right)$ will remain positive, and (57) is still kept in the form of (48). By using \emph{Lemma 3}, the finite time converge of $e_1$ is guaranteed, which also means that $e_1$ can converge to the region $\left| {{e_1}} \right| \le {\left( {\Delta_1 /{\gamma _5}} \right)^{1/{p_5}}}$, i.e. the tracking error $e_1$ will ultimately converge to the region $\left| {{e_1}} \right| \le {\left( {\Delta_1 /{\gamma _5}} \right)^{1/{p_5}}}$ in finite time. The proof is completed.
\end{proof}

%3-D
\subsection{The design of an observer}
For the elevation control channel
\begin{equation}
{\dot e_2} = {g_1}{v_1} + {f_1} + {d_1} 
\end{equation}
where ${d_1}$ is a time-varying lumped disturbance satisfying \emph{Assumption 1}. Then we will develop a novel sliding mode observer to estimate the disturbance based on \emph{Proposition 1}, which is call adaptive-gain smooth second-order sliding mode observer (ASSOSMO). 

First, an auxiliary estimation system is defined as
\begin{equation}
{\dot \hat e_2} = {g_1}{v_1} + {f_1} + {\hat d_1} 
\end{equation}

Then a sliding mode variable is defined as
\begin{equation}
{s_d} = {e_2} - {\hat e_2}
\end{equation}

Finally, the ASSOSMO can be designed as
\begin{equation}\small
\begin{aligned}
&{{\hat d}_1}{\rm{ = }}{L_{d1}}(t){\left| {{s_d}} \right|^{\frac{{m - 1}}{m}}}{\mathop{\rm sgn}} ({s_d}) + {L_{d2}}(t){s_d} + {\varphi _d}\\
&{{\dot \varphi }_d} = {L_{d3}}(t){\left| {{s_d}} \right|^{\frac{{m - 2}}{m}}}{\mathop{\rm sgn}} ({s_d}) + {L_{d4}}(t){s_d}
\end{aligned}
\end{equation}
where ${L_{1d}}(t),{L_{2d}}(t),{L_{3d}}(t),{L_{4d}}(t)$ are formulated the same as (15) and $m>2$.

\begin{theorem}
Under the condition of \emph{Assumption 1}, the proposed observer (61) can guarantee that the estimation disturbance  ${\hat d_1}$ converges to a neighborhood of the true disturbance ${d_1}$.
\end{theorem}
\begin{proof}
Taking the derivative of (60), we can obtain
\begin{equation}
{\dot s_d} = {\dot e_2} - {\dot \hat e_2} = {d_1} - {\hat d_1}
\end{equation}

Substituting (61) into (62) leads to
\begin{equation}
\begin{aligned}
&{{\dot s}_d} =  - {L_{d1}}(t){\left| {{s_d}} \right|^{\frac{{m - 1}}{m}}}{\mathop{\rm sgn}} ({s_d}) - {L_{d2}}(t){s_d} + {\varphi _{d1}}\\
&{{\dot \varphi }_{d1}} =  - {L_{d3}}(t){\left| {{s_d}} \right|^{\frac{{m - 2}}{m}}}{\mathop{\rm sgn}} ({s_d}) - {L_{d4}}(t){s_d} + {{\dot d}_1}
\end{aligned}
\end{equation}

In terms of \emph{Assumption 1}, the disturbance $\left| {{{\dot d}_1}} \right| \le {\delta _1}$. Meanwhile, the ${L_{1d}}(t),{L_{2d}}(t),{L_{3d}}(t),{L_{4d}}(t)$ are formulated the same as (15) and $m>2$. By using the \emph{ Proposition 1}, we can get that ${\dot s_d}$ converges to a neighborhood of the origin. Then by using (62), the estimation disturbance ${\hat d_1}$ can converge to a neighborhood of the true disturbance ${d_1}$, and the proof is accomplished.

\begin{remark}
The ASOSMO can also be developed following the same process by using the method in [16], which is selected as the comparing method. The effectiveness and smoothness of our proposed observer will be verified by the comparative simulation experiment.
\end{remark}

\end{proof}
%4-section
\section{Experimental Results}
This section presents three sets of comparative simulation experiments to expound the effectiveness of our proposed control schemes. The AFSOSMC method proposed in [16], combined with an integral non-singular terminal sliding mode (INTSM) surface, is implemented as the comparison method, where the whole control scheme is proposed and analyzed in \cite{Tran2020}. The INTSM-AFSOSMC method can be specifically expressed as follows

(\romannumeral1) An INTSM surface is selected as
\begin{equation}
{s_{ci}} = {e_{i + 1}} + \int_0^t {{{\left| {{z_{ci}}} \right|}^{\frac{1}{3}}}{\mathop{\rm sgn}} ({z_{ci}})d\tau}  
\end{equation}
where ${z_{ci}} = {e_i} + {\left| {{e_{i + 1}}} \right|^{\frac{3}{2}}}{\mathop{\rm sgn}} ({e_{i + 1}})$

(\romannumeral2) The INTSM-AFSOSMC law is
\begin{equation}
\begin{aligned}
u_{ci} = &{g_i}^{ - 1}\left\{ { - {L_{c1}}(t){{\left| {{s_{ci}}} \right|}^{\frac{1}{2}}}{\mathop{\rm sgn}} ({s_{ci}}) - {L_{c2}}(t){s_{ci}} - {{\left| {{z_{ci}}} \right|}^{\frac{1}{3}}}{\mathop{\rm sgn}} ({z_{ci}})}\right.\\
 &\phantom{=\;\;}\left.{ - {f_i} - \int_0^t {\left[ {{L_{c3}}(t){\mathop{\rm sgn}} ({s_{ci}}) + {L_{c4}}(t){s_{ci}}} \right]d\tau } } \right\}
\end{aligned}
\end{equation}
where $i=1,3$  and the time-varying gains are defined the same as (15) with $m=2$.

Table \ref{table:1} shows the parameter values of the helicopter system, and the fixed sampling time for all the simulation experiments is set to 0.001 second. In these experiments, the same initial position of the system starts from the elevation angle $-24^o$ and pitch angle $0^o$. The desired trajectories of elevation angle and pitch angle are given as
\begin{equation}
{x_{\alpha d}}(t){\rm{ = }}0.2\sin (0.08t - \frac{\pi }{2}),{x_{\beta d}}(t) = 0.1\sin (0.06t) 
\end{equation}

The purposes of three group experiments are summarized as follows
\begin{enumerate}[1)]
	\item The first experiment shows the advantages of our control scheme for the attitude tracking of the helicopter system with constant disturbance by comparing with the control method in \cite{Tran2020}. 
    \item The second experiment demonstrates the merits of our control scheme for the attitude tracking of the helicopter system with time-varying disturbance, which is also compared with the control method in \cite{Tran2020}.
    \item The third experiment is used to illustrate the smoother output of our proposed observer compared with the observer proposed from [16] in terms of estimating the time-varying lumped disturbances of the helicopter system.
\end{enumerate}
%4-A
\subsection{Experiment \uppercase\expandafter{\romannumeral1}: Attitude tracking control with constant disturbance}
For experiment \uppercase\expandafter{\romannumeral1}, the parameters of our proposed control scheme for elevation angle and pitch angle tracking are set as: ${\gamma _1} = {\gamma _2} = 5$, ${p_3} = 0.6$, $m = 3$, ${k_1} = 2$, ${k_2}$ = 2.5, ${k_3} = 4$, ${k_4} = 30$, $\kappa  = 5$. The INTSM-AFSOSMC law adopts the same gain parameters.
\begin{figure}
	\centering
	\subfigure[Response of elevation tracking error by AFSOSMC and proposed method]{
	\includegraphics[width=0.4\textwidth]{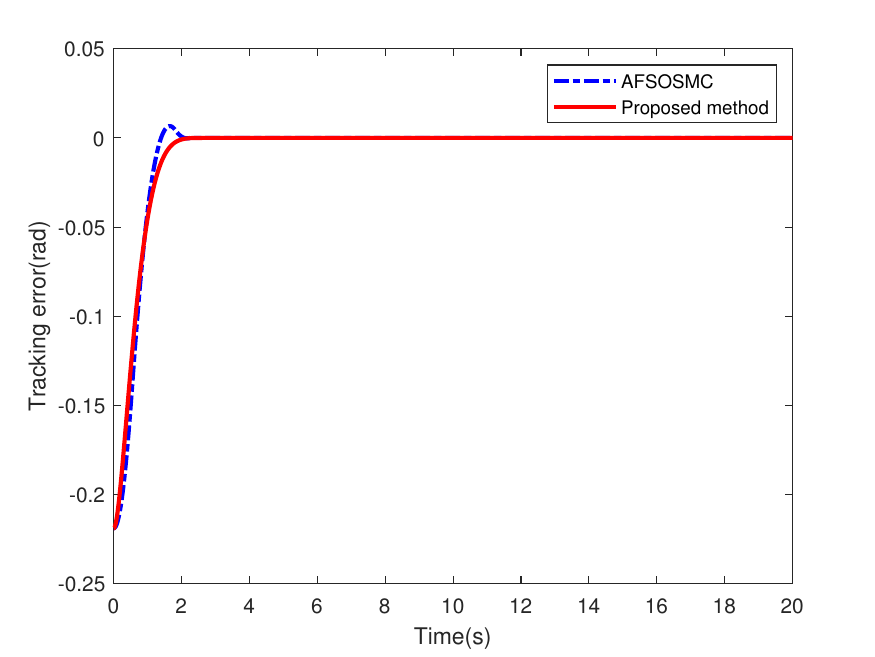}}
	\subfigure[Local magnification of elevation tracking error]{
	\includegraphics[width=0.4\textwidth]{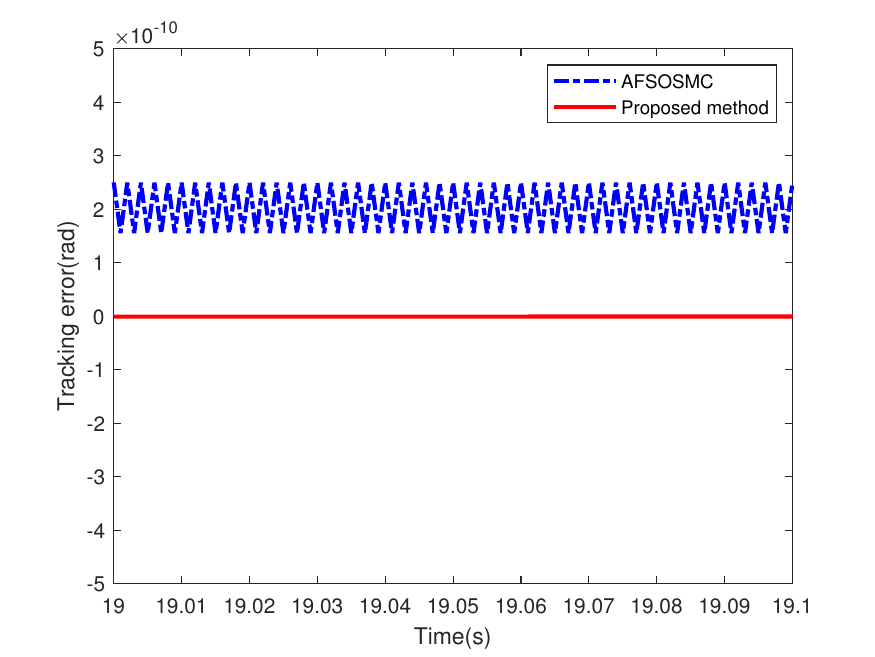}}
	\subfigure[Response of pitch tracking error by AFSOSMC and proposed method]{
	\includegraphics[width=0.4\textwidth]{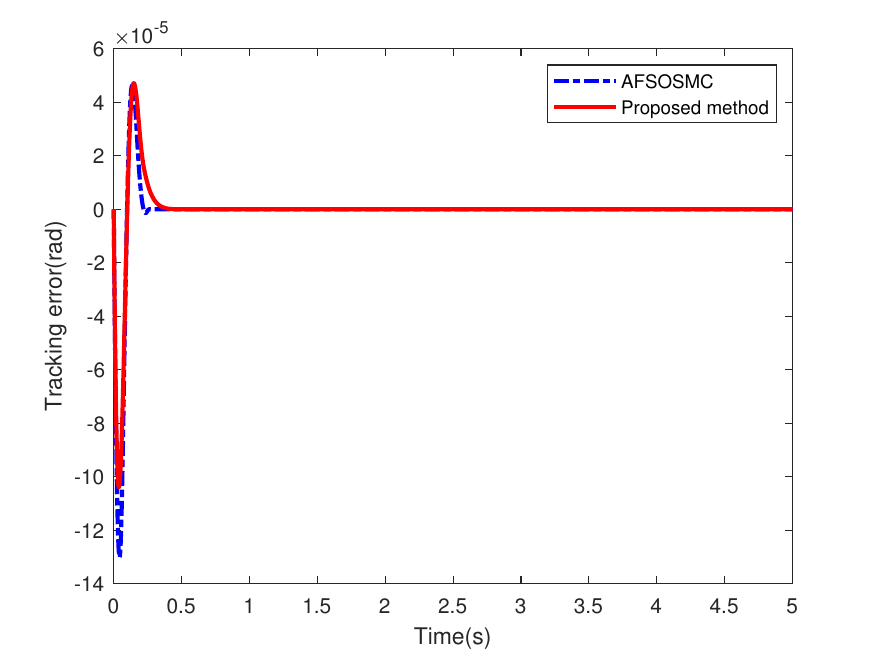}}
	\subfigure[Local magnification of pitch tracking error]{
	\includegraphics[width=0.4\textwidth]{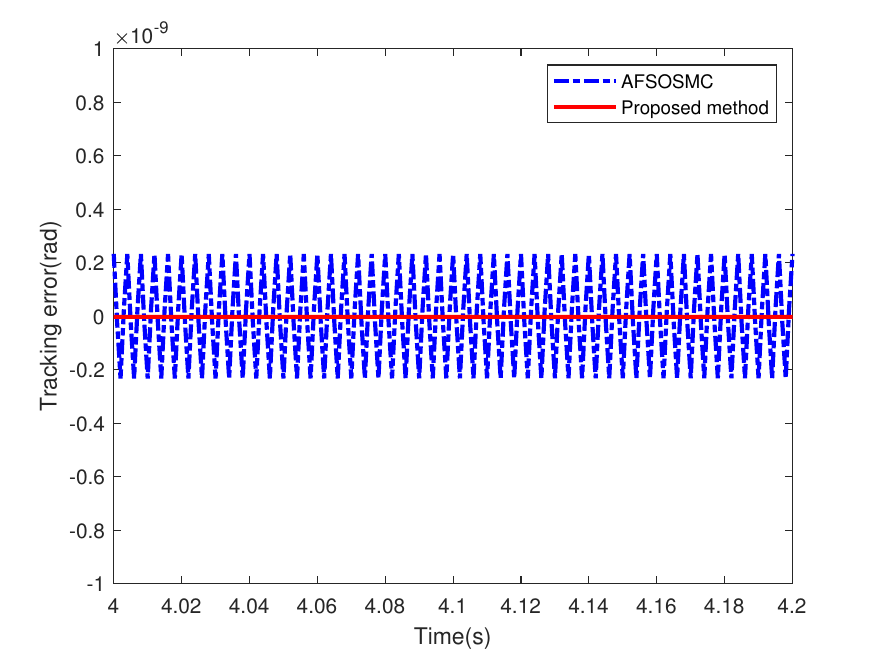}}
	\caption{Results of experiment \uppercase\expandafter{\romannumeral1}}
\end{figure}

Fig. 2: (a)-(d) show the experimental results. Fig. 2: (a) and (c) demonstrate the responses of elevation angle tracking error by using INTSM-AFSOSMC and the proposed method, respectively. Fig. 2: (b) and (d) present the local magnification of corresponding tracking error to show the steady-state response more clearly. Fig. 2: (a) shows that elevation angle and pitch angel tracking errors converge to the origin in finite time by using our proposed control scheme as fast as the method in \cite{Tran2020}, which means our proposed method also possesses the characteristic of fast finite-time convergence. In addition, Fig. 2: (b) illustrates that the proposed method can weaken the chattering effect existing in the AFSOSMC method. A similar conclusion for the pitch angel tracking can be drawn in light of the Fig. 2: (c) and (d).

%4-B
\subsection{Experiment \uppercase\expandafter{\romannumeral2}: Attitude tracking control with time-varying lumped disturbances}
For experiment \uppercase\expandafter{\romannumeral2}, the parameters of the both controller are as same as those in experiment \uppercase\expandafter{\romannumeral1}. The time-varying lumped disturbances are set as ${d_1}(t) = {d_2}(t) = 0.2\sin (t)$.
\begin{figure}
	\centering
	\subfigure[Response of elevation tracking error by AFSOSMC and proposed method]{
	\includegraphics[width=0.4\textwidth]{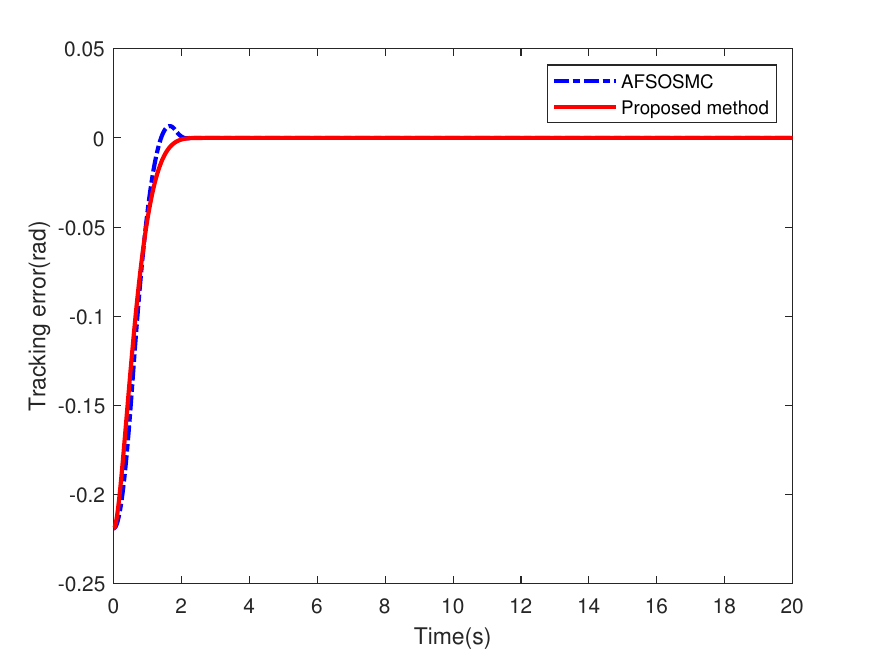}}
	\subfigure[Local magnification of elevation tracking error]{
	\includegraphics[width=0.4\textwidth]{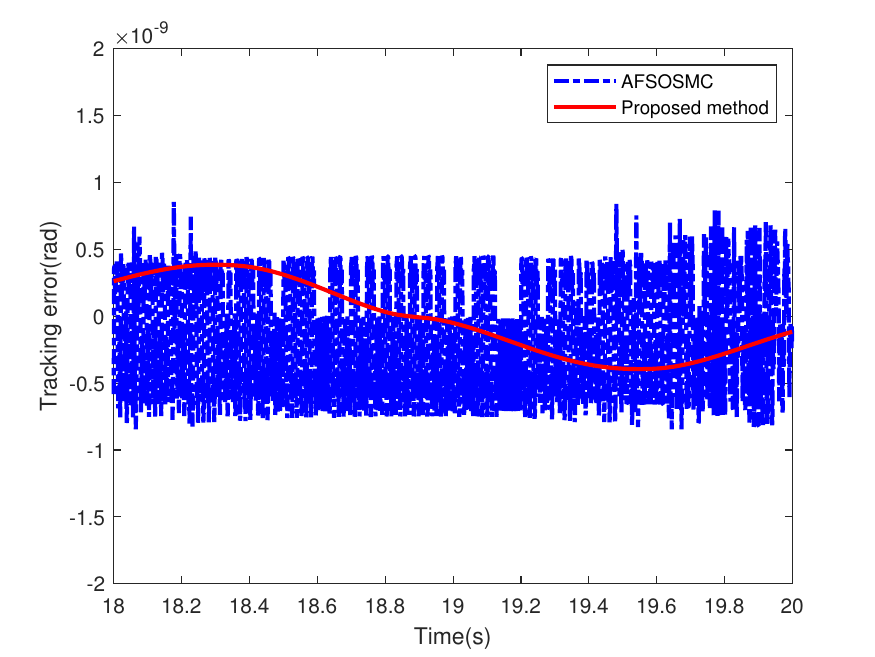}}
	\subfigure[Response of pitch tracking error by AFSOSMC and proposed method]{
	\includegraphics[width=0.4\textwidth]{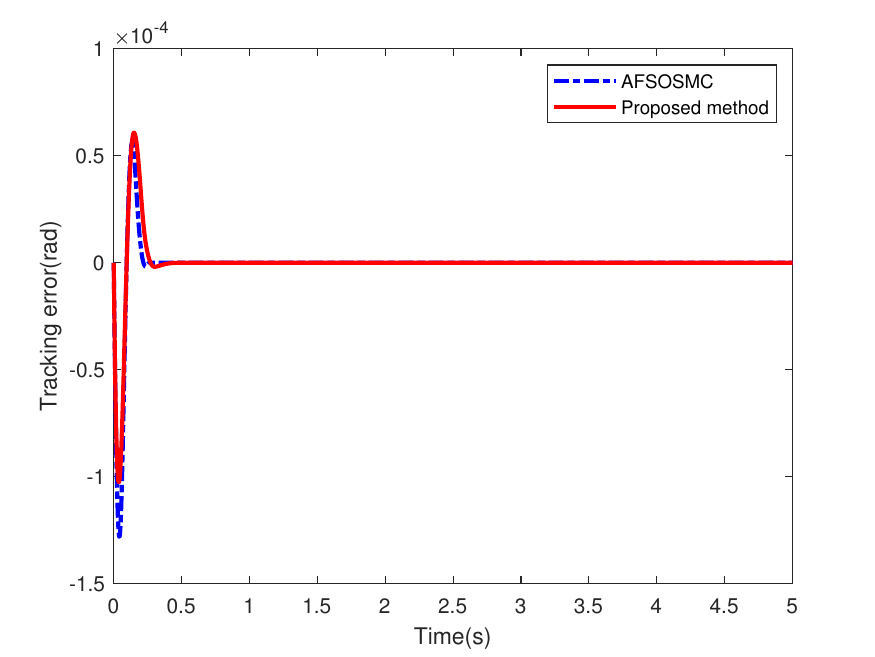}}
	\subfigure[Local magnification of pitch tracking error]{
	\includegraphics[width=0.4\textwidth]{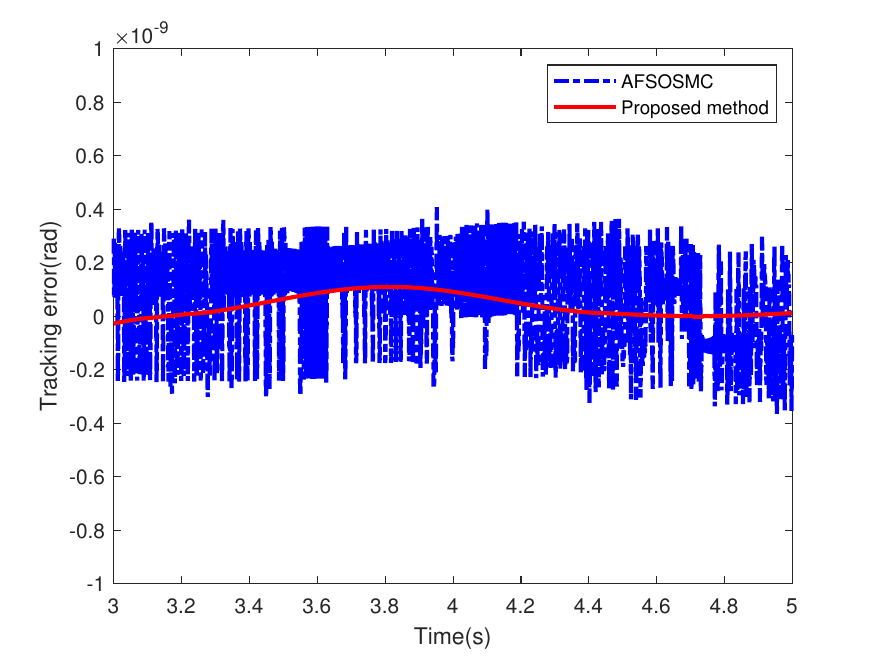}}
	\caption{Results of experiment \uppercase\expandafter{\romannumeral2}}
\end{figure}

Fig. 3: (a)-(d) show the experimental results. Fig. 3: (a) and (c) illustrate the responses of elevation angle tracking error by using INTSM-AFSOSMC and the proposed method, respectively. Fig. 3: (b) and (d) present the local magnification of corresponding tracking error to show the steady-state response more clearly. From Fig. 3: (a) and (b), one can see that the proposed method can guarantee the fast finite-time uniformly ultimately boundedness coupled with large chattering suppression. However, the time-varying disturbance destroys the original property of finite-time convergence, which leads to the tracking errors only converging in a small neighborhood of the origin.  A similar conclusion for the pitch angel tracking can be drawn in light of the Fig. 3: (c) and (d).

%4-C
\subsection{Experiment \uppercase\expandafter{\romannumeral3}: Performance analysis of the observer}
For experiment \uppercase\expandafter{\romannumeral3}, we compare the performance of our proposed observer with the  ASOSMO. The parameters of the proposed observer are set as $m = 3,{k_1} = 2,{k_2} = 2.5,{k_3} = 4,{k_4} = 30,\kappa  = 10$, and the ASOSMO adopts the same parameters except that $m$ is set as $2$. To facilitate the comparison, we take the elevation angel tracking as an example and the time-varying disturbances is set as $\sin (t)$. Moreover, the following conventional sliding mode control method is selected as the control law
\begin{equation}
{u_t} = {g_1}^{ - 1}\left( { - {f_1} - {c_t}{e_2} - \eta {\mathop{\rm sgn}} ({s_t}) - {k_t}{s_t}} \right)
\end{equation}
where the sliding mode surface ${s_t} = {e_2} + {c_t}{e_1}$ and ${c_t} = 2,{k_t} = 0.5$, $\eta $ is set sufficiently large to ensure the stability of the helicopter system.
\begin{figure}
	\centering
	\subfigure[Disturbance estimation error by using ASOSMO and ASSOSMO]{
	\includegraphics[width=0.4\textwidth]{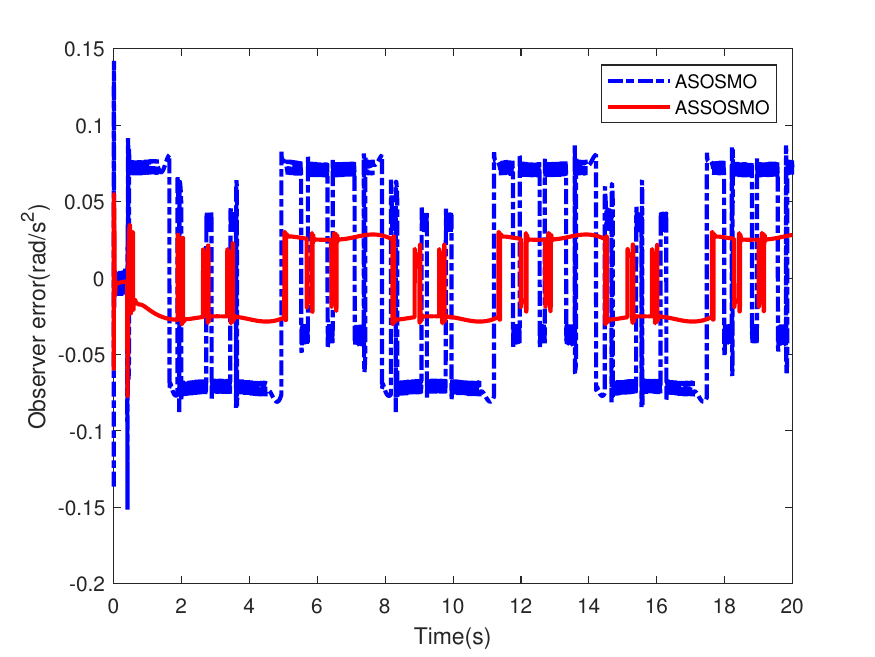}}
	\subfigure[Local magnification of disturbance estimation error]{
	\includegraphics[width=0.4\textwidth]{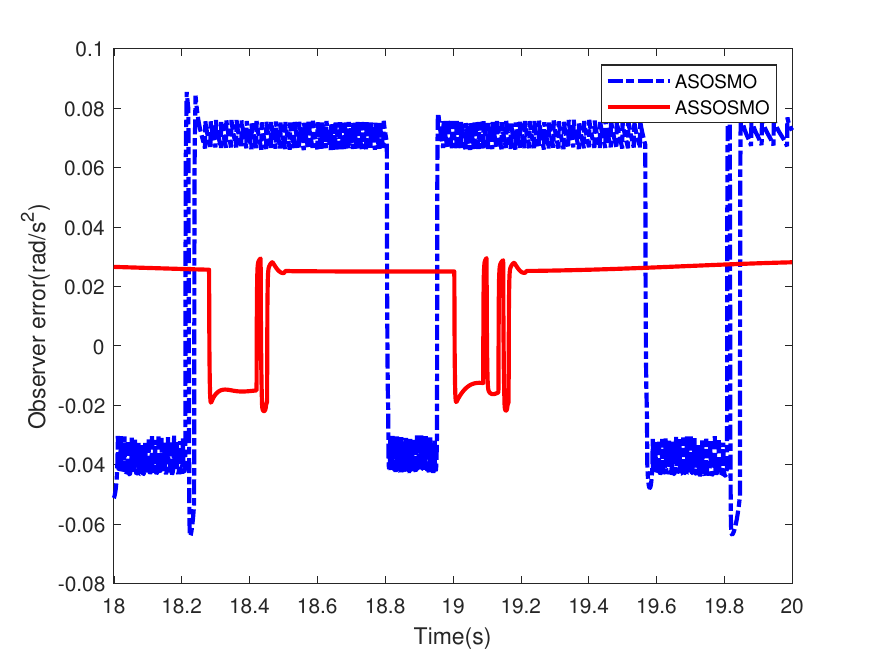}}
	\caption{Results of experiment \uppercase\expandafter{\romannumeral3}}
\end{figure}

Fig. 4: (a)-(b) show the experimental results. Fig. 4: (a) illustrates the response of observer estimation error by using ASOSMO and the proposed observer and Fig. 4: (b) present the local magnification to show the response more clearly. From Fig. 4: (a) and (b), one can see that the estimation error can fast converge to a region of the origin in finite time by using our proposed observer. Furthermore, the proposed observer can effectively attenuate chattering existing in the ASOSMO and provide a more accurate and smoother output for the disturbance estimation.

\section{Conclusion}
In this paper, a novel smooth adaptive fast second-order sliding mode control algorithm has been proposed. According to the types of disturbances, the design process of helicopter attitude tracking controllers can be divided into two parts. For the system with constant disturbances, a control scheme, consisting of a non-singular integral sliding mode surface and a novel smooth adaptive fast second-order sliding mode control, is proposed to achieve the fast finite-time convergence. For the system with the time-varying lumped disturbances, the control scheme can realize the fast finite-time uniformly ultimately boundedness. A novel adaptive-gain smooth second-order sliding mode observer is also present to estimate disturbances with a smooth output. The comparative simulation experiments are performed to demonstrate the effectiveness and superiority of the proposed control scheme with fast finite-time convergence, adaptation to disturbances, and chattering suppression for the attitude tracking of the 3-DOF helicopter system.

%\section*{Acknowledgment}

% References
\bibliographystyle{Bibliography/IEEEtranTIE}
\bibliography{Bibliography/IEEEabrv,Bibliography/myRef}\ %IEEEabrv instead of IEEEfull

\end{document}